\documentclass[11pt,a4paper]{article}
\usepackage[a4paper,margin=2cm]{geometry}
\usepackage[utf8]{inputenc}
\usepackage{mathtools}
\usepackage{amsmath}
\usepackage{amsthm}
\usepackage{color}
\usepackage{graphicx}
\usepackage{booktabs}
\usepackage{hyperref}

\usepackage{enumitem,amssymb}
\usepackage{authblk}
\usepackage{natbib}

\newtheorem{theorem}{Theorem}

\definecolor{myorange}{RGB}{255,165,0}
\definecolor{mygreen}{RGB}{46, 139, 87}
\newcommand{\winner}[1]{\textcolor{blue}{\textbf{#1}}}

\title{Manifold-adaptive dimension estimation revisited}

\author[1,2, *]{Zsigmond Benkő}
\author[1]{Marcell Stippinger}
\author[1]{Roberta Rehus}
\author[1]{Attila Bencze}
\author[4]{Dániel Fabó}
\author[2, 4]{Boglárka Hajnal}
\author[5,6]{Loránd Erőss}
\author[1,3,7]{András Telcs}
\author[1,8]{Zoltán Somogyvári}

\affil[1]{Department of Computational Sciences, Wigner Research Centre for Physics,  H-1121, Hungary}
\affil[2]{János Szentágothai Doctoral School of Neurosciences, Semmelweis University, H-1085, Hungary}
\affil[3]{Department of Computer Science and Information Theory, Faculty of Electrical Engineering and Informatics, Budapest University of Technology and Economics, Budapest, H-1111, Hungary}
\affil[4]{Epilepsy Center, Department of Neurology, National Institute of Clinical Neurosciences, Budapest, H-1145, Hungary}
\affil[5]{Department of Functional Neurosurgery, National Institute of Clinical Neurosciences, Budapest, H-1145, Hungary}
\affil[6]{Faculty of Information Technology and Bionics, Péter Pázmány Catholic University, Budapest, H-1083, Hungary}
\affil[7]{Department of Quantitative Methods, University of Pannonia, Faculty of Business and Economics, H-8200, Veszprém, Hungary}
\affil[8]{Neuromicrosystems ltd., Budapest, H-1113, Hungary}
\affil[*]{benko.zsigmond@wigner.hu}

\begin{document}

\maketitle

\begin{abstract}
   Data dimensionality informs us about data complexity and sets limit on the structure of successful signal processing pipelines.
    In this work we revisit and improve the manifold adaptive Farahmand-Szepesvári-Audibert (FSA) dimension estimator, making it one of the best nearest neighbor-based dimension estimators available.
    We compute the probability density function of local FSA estimates, if the local manifold density is uniform.
    Based on the probability density function, we propose to use the median of local estimates as a basic global measure of intrinsic dimensionality, and we demonstrate the advantages of this asymptotically unbiased estimator over the previously proposed statistics: the mode and the mean.
    Additionally, from the probability density function, we derive the maximum likelihood formula for global intrinsic dimensionality, if i.i.d. holds.
    We tackle edge and finite-sample effects with an exponential correction formula, calibrated on hypercube datasets.
    We compare the performance of the corrected-median-FSA estimator with kNN estimators: maximum likelihood (ML, Levina-Bickel) and two implementations of DANCo (R and matlab).
    We show that corrected-median-FSA estimator beats the ML estimator and it is on equal footing with DANCo for standard synthetic benchmarks according to mean percentage error and error rate metrics.
    With the median-FSA algorithm, we reveal  diverse changes in the neural dynamics while resting state and during epileptic seizures.
    We identify brain areas with lower-dimensional dynamics that are possible causal sources and candidates for being seizure onset zones.
\end{abstract}

\section*{Introduction}

Dimensionality sets profound limits on the stage where data takes place, therefore it is often crucial to know the intrinsic dimension of data to carry out meaningful analysis.
Intrinsic dimension provides direct information about data complexity, as such, it was recognised as a useful measure to describe the dynamics of dynamical systems\citep{Grassberger1983}, to detect anomalies in time series\citep{Houle2018177}, to diagnose patients with various conditions\citep{Dlask2017, Polychronaki2010, Sharma2017, Acharya2013} and to use it simply as plugin parameter for signal processing algorithms. 

Most of the multivariate datasets lie on a lower dimensional manifold embedded in a potentially very high-dimensional embedding space.
This is because the observed variables are far from independent, and this interdependence introduces redundancies resulting in a lower intrinsic dimension (ID) of data compared with the number of observed variables.
To capture this -- possibly non-linear -- interdependence,  nonlinear dimension-estimation techniques can be applied\citep{Sugiyama2013, Romano2016, Benko2018, Krakovska2019}.

To estimate the ID of data various aproaches have been proposed, for a full review of techniques see the work of Campadelli et al.\citep{campadelli2015intrinsic}.
Here we discuss the k-Nearest Neighbor (kNN) ID estimators, with  some recent advancements in the focus.

A usually basic assumption of $k$NN ID estimators is that the fraction of points in a neighborhood is approximately determined by the intrinsic dimensionality ($D$) and distance ($R$) times a -- locally almost constant -- mostly density-dependent factor ($\eta(x, R)$, Eq.\,\ref{eq:knn}).

\begin{equation}\label{eq:knn}
    \frac{k}{n} \approx \eta(x, R) *  R_k^D
\end{equation}
where $k$ is the number of samples in a neighborhood and $n$ is the total number of samples on the manifold.

Assuming a Poisson sampling process on the manifold Levina and Bickel\citep{Levina2005} derived a Maximum Likelihood estimator, which became a popular method and got several updates\citep{Ghahramani2005, Gupta2010}.
These estimators are prone to underestimation of dimensionality because of finite sample effects and overestimations because of the curvature.

To address the challenges posed by curvature and finite sample, new estimators were proposed \citep{Rozza2012, Bassis2015, Ceruti2014, Facco2017}.
To tackle the effect of curvature, a minimal neighborhood size can be taken on normalized neighborhood distances as in the case of $\mathrm{MIND}_{\mathrm{ML}}$\citep{Rozza2012}.
To tackle the underestimation due to finite sample effects, empirical corrections were applied.
A naive empirical correction approach was applied by Camastra and Vinciarelli\citep{Camastra2002}: a perceptron was trained on the estimates computed for  randomly sampled hypercubes to learn a correction function.
Motivated by the correction in the previous work, the IDEA method was created\citep{Rozza2012}; and a more principled approach was carried out, where the full distribution of estimates was compared to the distributions computed on test data sets using the Kullback-Leibner divergence ($\mathrm{MIND}_{\mathrm{KL}}$\citep{Rozza2012}, DANCo\citep{Ceruti2014}).
In the case of DANCo, not just the nearest neighbor distances, but the angles are measured and taken into account in the estimation process resulting in more accurate estimates.

In the recent years, further estimators have been proposed, such as the estimator that uses minimal neighborhood information leveraging the empirical distribution of the ratio of the nearest neighbors to fit intrinsic dimension\citep{Facco2017}, or other approaches based on simplex skewness\citep{Johnsson2015} and normalized distances \citep{Chelly2016, Amsaleg2015, Amsaleg2018, Amsaleg2019}.

In the following section, we revisit the manifold adaptive dimension estimator
 proposed by Farahmand et al.\citep{Farahmand2007} to measure intrinsic dimensionality of datasets.
From Eq. \ref{eq:knn} we can take the logarithm of both sides:

\begin{equation}
    \begin{split}
        \ln \left( \frac{k}{n} \right) &\approx \ln{\eta} + D \ln{R_k}\\
        \ln \left( \frac{2k}{n} \right) &\approx \ln{\eta} + D \ln{R_{2k}}
    \end{split}
\end{equation}
If $\eta$ is slowly varying and $R$ is small, we can take it as a constant.

If we subtract the two equations from each other we get:

\begin{equation}
     \ln{\left( 2 \right)} \approx D \ln{\left( \frac{R_{2k}}{R_k} \right)}
\end{equation}

Thus, to compute the local estimates, we fit a line through the log-distance $k$th and $2k$th nearest neighbor at a given location.

\begin{equation}
    d(x) = \frac{\ln(2)}{\ln \left( R_{2k}/R_{k} \right)}
\end{equation}

To compute a global ID estimate, the authors proposed the mean of local estimates at sample-points, or a vote for the winner global ID value (the mode), if the estimator is used in integer-mode.
They proved that the above global ID estimates are consistent for $k>1$, if $\eta$ is differentiable and the manifold is regular.
They calculated the upper bound for the probability of error for the global estimate, however this bound contains unknown constants\citep{Farahmand2007}.

In this paper we propose an improved FSA estimator, based on the assumption that the density is locally uniform.
We suggest to use the median of local values for a global intrinsic dimension estimate.
We correct the underestimation effect by an exponential formula and test the new algorithm on benchmark datasets.
We apply the proposed estimator to locate epileptic focus on field potential measurements.

\begin{figure}[htb!]
    \centering
    \includegraphics[width=0.9\linewidth]{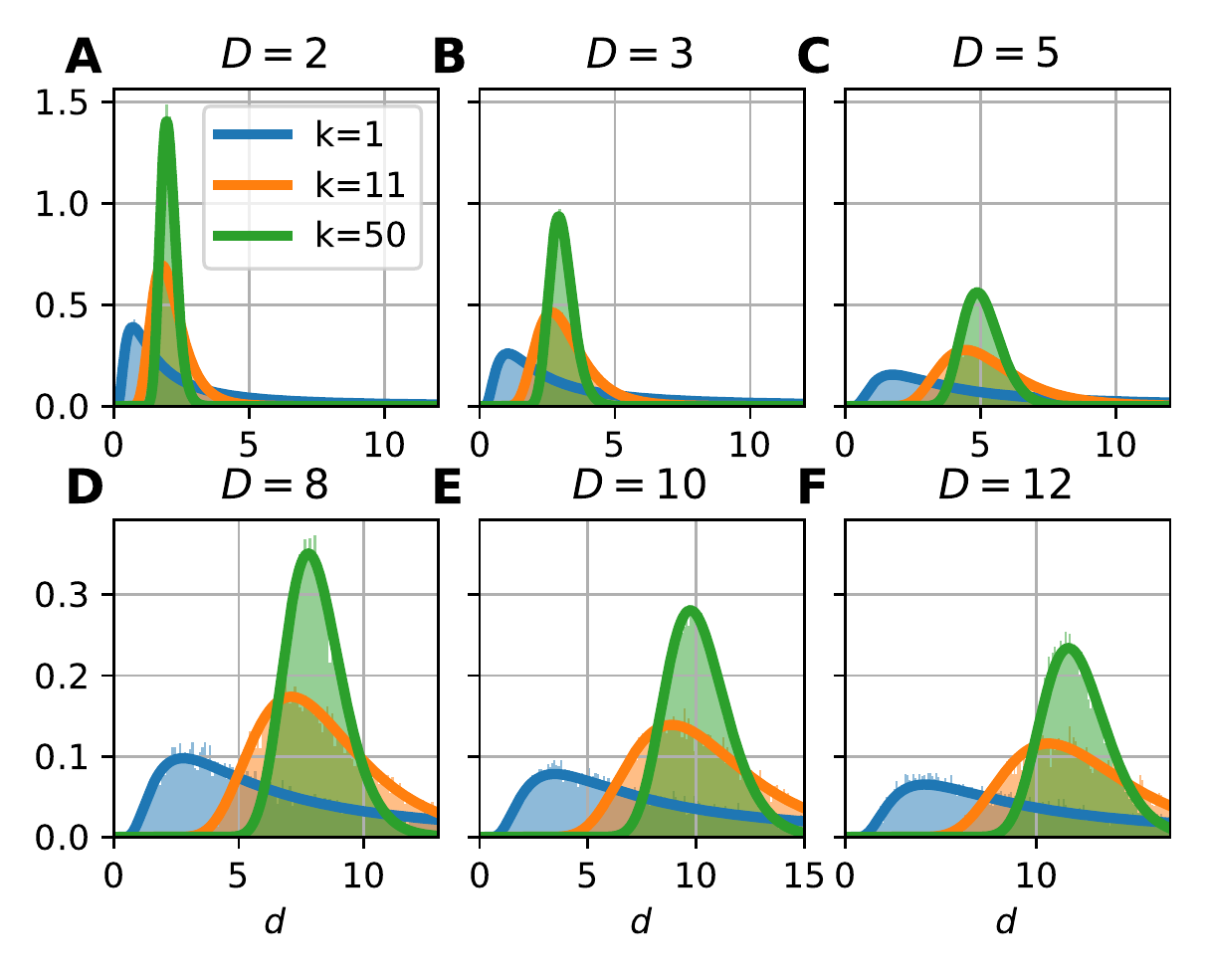}
    \caption{\textbf{Probability density function of the Farahmand-Szepesvári-Audibert estimator ($d$) for various dimensions ($D$) and neighborhood sizes ($k$).} \textbf{A-F} The sublots show that theoretical pdfs (continuous lines) fits to the histograms ($n=10000$) of local estimates calculated on uniformly sampled hypercubes ($D=2, 3, 5, 8, 10, 12$). The three colors denote three presented neighborhood sizes: $k=1$ (blue), $k=11$ (orange) and $k=50$ (green).}
    \label{fig:szepes_pdf}
\end{figure}

\section*{Results}

\subsection*{Manifold adaptive dimension estimator revisited}
\subsubsection*{The probability density of Farahmand-Szepesvári-Audibert estimator}
We compute the probability density function of Farahmand-Szepesvári-Audibert (FSA) intrinsic dimension estimator based on normalized distances.

The normalized distance density of the $k$NN can be computed in the context of a $K$-neighborhood, where the normalized distance of $K-1$ points follows a specific form:

\begin{equation}\label{eq:pdf_r}
    p(r|k, K-1, D) = \frac{D}{B(k, K-k)} r ^ {D k -1} (1 - r^D) ^ {K-k-1}
\end{equation}
where $r$ is the normalized distance of the $k$th neighbor by the distance of $K$th neighbor ($r_k = R_k / R_K$, $k<K$) and $B$ is the Euler-beta function (see SI\,\ref{si:r_pdf} for a derivation).
A maximum likelihood estimator based on Eq.\,\ref{eq:pdf_r} leads to the formula of the classical Levina-Bickel estimator (\citep{Levina2005}). 
For a derivation of this probability density and the maximum likelihood solution see SI\,\ref{si:r_pdf} and SI\,\ref{si:ML} respectively.

We realize that the inverse of normalized distance appears in the formula of FSA estimator, so we can express it as a function of $r$:

\begin{equation}
    d_k = \frac{\log{2}}{\log{\left( R_{2k} / R_{k} \right)} }
     = -\frac{\log{2}}{\log{\left( R_{k} / R_{2k} \right)} } 
 = -\frac{\log{2}}{\log{r_k}}
\end{equation}
Where $r_k = R_k / R_{2k}$.

Thus, we can compute the pdf of the estimated values as plugging in $K=2k$ into Eq.\,\ref{eq:pdf_r} followed by change of variables:

\begin{equation}\label{eq:szepes_pdf}
     q \left( d_k \right) \equiv p \left( r|k, 2k-1, D \right) \left| \frac{\mathrm{d}r}{\mathrm{d}d_k} \right| =  \frac{D \log{(2)}}{B(k, k)}   \frac{2^{-\frac{Dk}{d_k}} \left(1 - 2^{-\frac{D}{d_k} } \right)^{k-1}}{d_k^2}
\end{equation}

\begin{theorem}
    The median of $q(d_k)$ is at $D$.
\end{theorem}

\begin{proof}
    We apply the substitution $a=2^{-D / d_k}$ in Eq. \ref{eq:szepes_pdf} (Eq. \ref{eq:beta}):
    
    \begin{align}
        p(a) &= q(d_k) \left| \frac{\mathrm{d}d_k}{\mathrm{d}a} \right| =\\
        &= \frac{D \log{(2)}}{B(k, k)} \frac{a^k (1-a)^{k-1} \log^2{a}}{D^2 \log^2{2}} \frac{D \log{2}}{a \log^2{a}} \\
        &= \frac{1}{B(k, k)} a^{k-1} (1-a)^{k-1} \label{eq:beta}
    \end{align}
    The pdf in Eq.\ref{eq:beta} belongs to a beta distribution.
    The cumulative distribution function of this density is the regularized incomplete Beta function with $k$ as both parameters symmetrically.
    
    \begin{equation}\label{eq:Pa}
        P(a) = I_a(k, k)
    \end{equation} 
    The median of this distribution is at $a=\frac{1}{2}$, thus at $d_k=D$ since:
    \begin{eqnarray}
        a = 2^{-\frac{D}{d_k}} & = & \frac{1}{2}\\
        D & = & d_k \label{eq:szepes_median_general}
    \end{eqnarray}
\end{proof}
This means that the median of the FSA estimator is equal to the intrinsic dimension independent of neighborhood size, if the locally uniform point density assumption holds.
The sample median is a robust statistic, therefore we propose to use the sample median of local estimates as a global dimension estimate. 
We will call this modified method the median Farahmand-Szepesvári-Audibert (mFSA) estimator.

Let's see the form for the smallest possible neighborhood size: $k=1$ (Fig.\,\ref{fig:szepes_pdf}).
The pdf  for the estimator takes a simpler from (Eq. \ref{eq:szepes_pdf_k1}).

\begin{equation}\label{eq:szepes_pdf_k1}
    q(d|k=1, D) = D \log(2) \frac{2^{-\frac{D}{d_1}}}{d_1^2}
\end{equation}

Also, we can calculate the cumulative distribution function analytically (Eq. \ref{eq:szepes_cdf_k1}).

\begin{equation}\label{eq:szepes_cdf_k1}
    Q(d|k=1, D) = \int_0^{d_1} q(t|k=1, D) \quad \mathrm{d}t = 2^{-D/d_1}
\end{equation}

The expectation of $d_k$ diverges for $k=1$-- but not for $k>1$ -- although the median exists.

From Eq. \ref{eq:szepes_cdf_k1} the median is at $D$ (Eq. \ref{eq:szepes_median}). 
\begin{equation}\label{eq:szepes_median}
    Q(d_1=D) = 0.5
\end{equation}

\subsubsection*{Sampling distribution of the median}

We can easily compute the pdf of the sample median if an odd sample size is given ($n = 2l + 1$) and if sample points are drawn independently according to Eq. \ref{eq:szepes_pdf}.
Roughly half of the points have to be smaller, half of the points have to be bigger and one point has to be exactly at $m$ (Eq. \ref{eq:median}). 

\begin{equation}\label{eq:median}
    \begin{split}
        p(m | k, D, n) &=  \frac{1}{B(l+1, l+1)} \left[ P\left(a=2^{-D/m}\right) \left(1 - P\left(a = 2 ^{-D/m}\right) \right) \right] ^ {l} q(m)
    \end{split}
\end{equation}
Where $p(a)$ and $P(a)$ are the pdf and cdf of $a$ (Eq.\,\ref{eq:beta},\,\ref{eq:Pa}) and $q$ is the pdf of the FSA estimator (Fig. \ref{fig:median_pdf}).

\begin{figure}[htb!]
    \centering
    \includegraphics[width=0.9\textwidth]{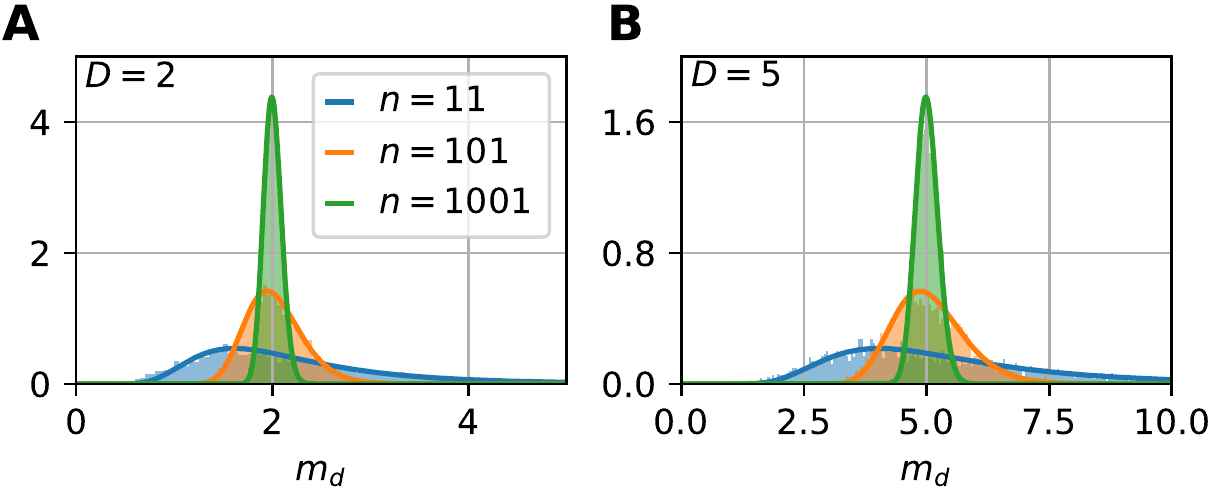}
    \caption{\textbf{The sampling distribution of the median for the FSA estimator ($k=1$) on uniformly sampled hypercubes.}
    The figure shows the pdf of median-FSA estimator of points uniformly sampled from a square (\textbf{A}) and from a 5D hypercube (\textbf{B}) for three sample sizes: $n=11$ (blue), $n=101$ (orange) and $n=1001$ (green) respectively. The solid lines represent the theoretical pdf-s of the median and the shaded histograms are the results of simulations ($N=5000$ realizations).
    }
    \label{fig:median_pdf}
\end{figure}

\subsubsection*{Maximum Likelihood solution for the manifold-adaptive estimator}

If the samples are independent and identically distributed, we can formulate the likelihood function as the product of sample-likelihoods (Eq. \ref{eq:FS_L}).
We seek for the maximum of the log likelihood function, but the derivative is transcendent for $k>1$.
Therefore, we can compute the place of the maximum numerically (Eq. \ref{eq:FS_ML}).

\begin{eqnarray}
    \mathcal{L} &=& \prod_{i=1}^{n} \frac{D \log{(2)}}{B(k, k)}  \frac{2 ^{-D k / {d_k}^{(i)}}  (1-2^{-D / {d_k}^{(i)}})^{k-1}}{{\left({d_k}^{(i)}\right)}^{2}} \label{eq:FS_L}\\
    \log \mathcal{L} &=& n \log \frac{\log{(2)}}{B(k, k)}  + n \log D - D k \log(2) \sum \frac{1}{{d_k}^{(i)}} + (k-1) \sum \log{ \left(1-2 ^ {-D / {d_k} ^ {(i)}} \right)} \\&&- 2 \sum \log({d_k}^{(i)})\nonumber\\
    \frac{\partial \log \mathcal{L}}{\partial D} &=& \frac{n}{D} - \log(2) k \sum \frac{1}{{d_k}^{(i)}} + \log(2) (k-1) \sum \frac{1}{{d_k}^{(i)} (2^{D/{d_k}^{(i)}} - 1)} \stackrel{!}{=} 0 \label{eq:FS_ML}
\end{eqnarray}

For $k=1$, the ML formula is equal to the Levina-Bickel ($k=1$) and $\mathrm{MIND}_{\mathrm{1ML}}$ formulas.

\subsection*{Results on randomly sampled hypercube datasets}
Theoretical probability density function of the local FSA estimator fits to empirical observations (Eq. \ref{eq:szepes_pdf}, Fig. \ref{fig:szepes_pdf}).
We simulated hypercube datasets with fixed sample size ($n=10000$) and of various intrinsic dimensions ($D=2, 3, 5, 8, 10, 12$).
We measured the local FSA estimator at each sample point with $3$ different $k$ parameter values ($k=1, 11, 50$).
We visually confirmed that the theoretical pdf fits perfectly to the empirical histograms.   

The empirical sampling distribution of mFSA fits to the theoretical curves for small intrinsic dimension values (Fig. \ref{fig:median_pdf}).
To demonstrate the fit, we drew the density of mFSA on two hypershpere datasets $D=2$ and $D=5$ with the smallest possible neighborhood ($k=1$), for different sample sizes ($n=11$, $101$, $1001$).
At big sample sizes the pdf is approximately a Gaussian\citep{Laplace1986}, but for small samples the pdf is non-Gaussian and skewed.

The mFSA estimator underestimates intrinsic dimensionality in high dimensions.
This phenomena is partially a finite sample effect (Fig.\,\ref{fig:szepes_ddep}), but edge effects make this underestimation even more severe.
We graphically showed that mFSA estimator asymptotically converged to the real dimension values for hypercube-datasets, when we applied periodic boundary conditions (Fig. \ref{fig:szepes_convergence}).
We found, that the convergence is much slower for hard boundary conditions, where edge effects make estimation errors higher.

We could estimate the logarithm of relative error with an $s$-order polynomial:

\begin{equation}
   \log(E_{rel}) = \log \left( \frac{D}{d} \right) = \sum_{i=1}^s \alpha_i d^i 
\end{equation}

The order of the polynomial was different for the two types of boundary conditions.
When we applied hard boundary, the order was $s=1$, however in the periodic case higher order polynomials fit the data. 
Thus, in the case of hard-boundary, we could make the empirical correction formula:
\begin{equation}
   D \approx C(\hat{d}) = d e^{\alpha_n d} 
\end{equation}
where $\alpha_n$ is a sample size dependent coefficient that we could fit with the least squares method.

\begin{figure}[t!]
    \centering
    \includegraphics[width=0.9\textwidth]{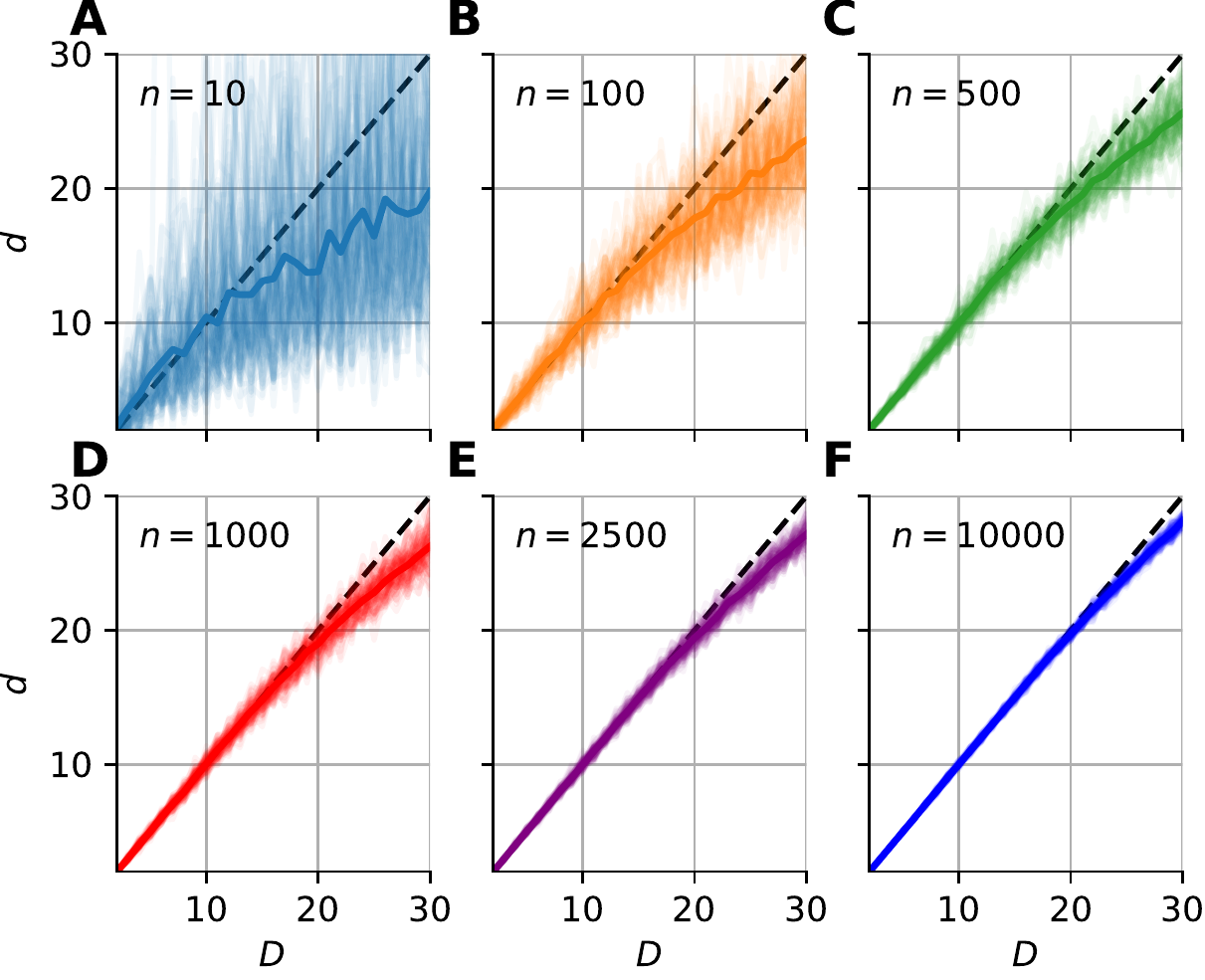}
    \caption{\textbf{Intrinsic dimension dependence of the median-FSA estimator for uniformly sampled unit hypercubes with various sample sizes ($k=1$).} Subplots \textbf{A-F} show the mean of median-FSA estimator (thick line) values from $N=100$ realizations (shading) of uniformly sampled unit hypercubes with periodic boundary. }
    \label{fig:szepes_ddep}
\end{figure}

\begin{figure}[t!]
    \centering
    \includegraphics[width=0.9\textwidth]{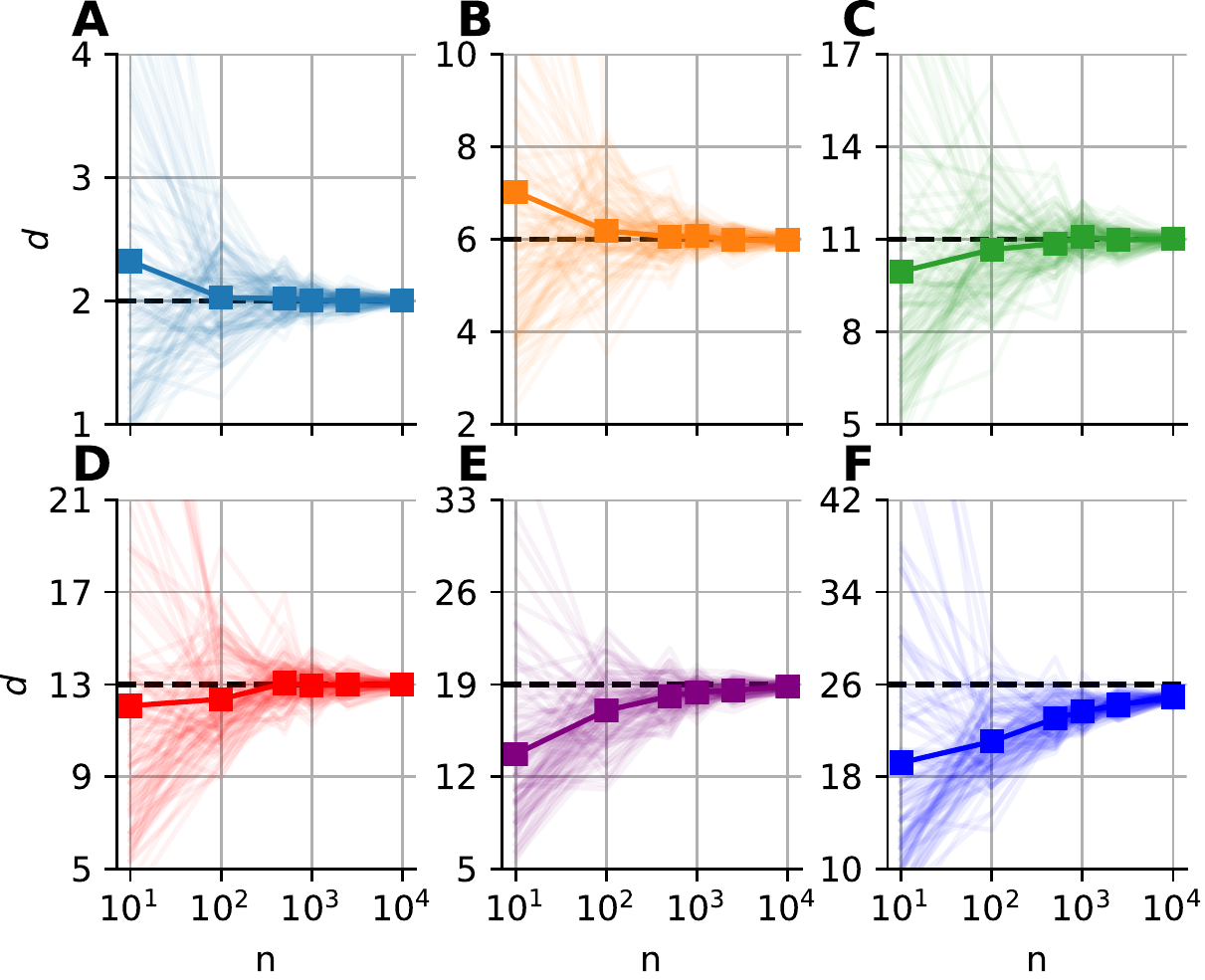}
    \caption{\textbf{Sample size dependence of the median-FSA estimator for uniformly sampled unit hypercubes with varied intrinsic dimension value ($k=1$).} Subplots \textbf{A-F} show the mean of median-FSA estimator (thick line) values from $N=100$ realizations (shading). 
    }
    \label{fig:szepes_convergence}
\end{figure}

\begin{figure}[htb!]
    \centering
    \includegraphics[width=0.9\textwidth]{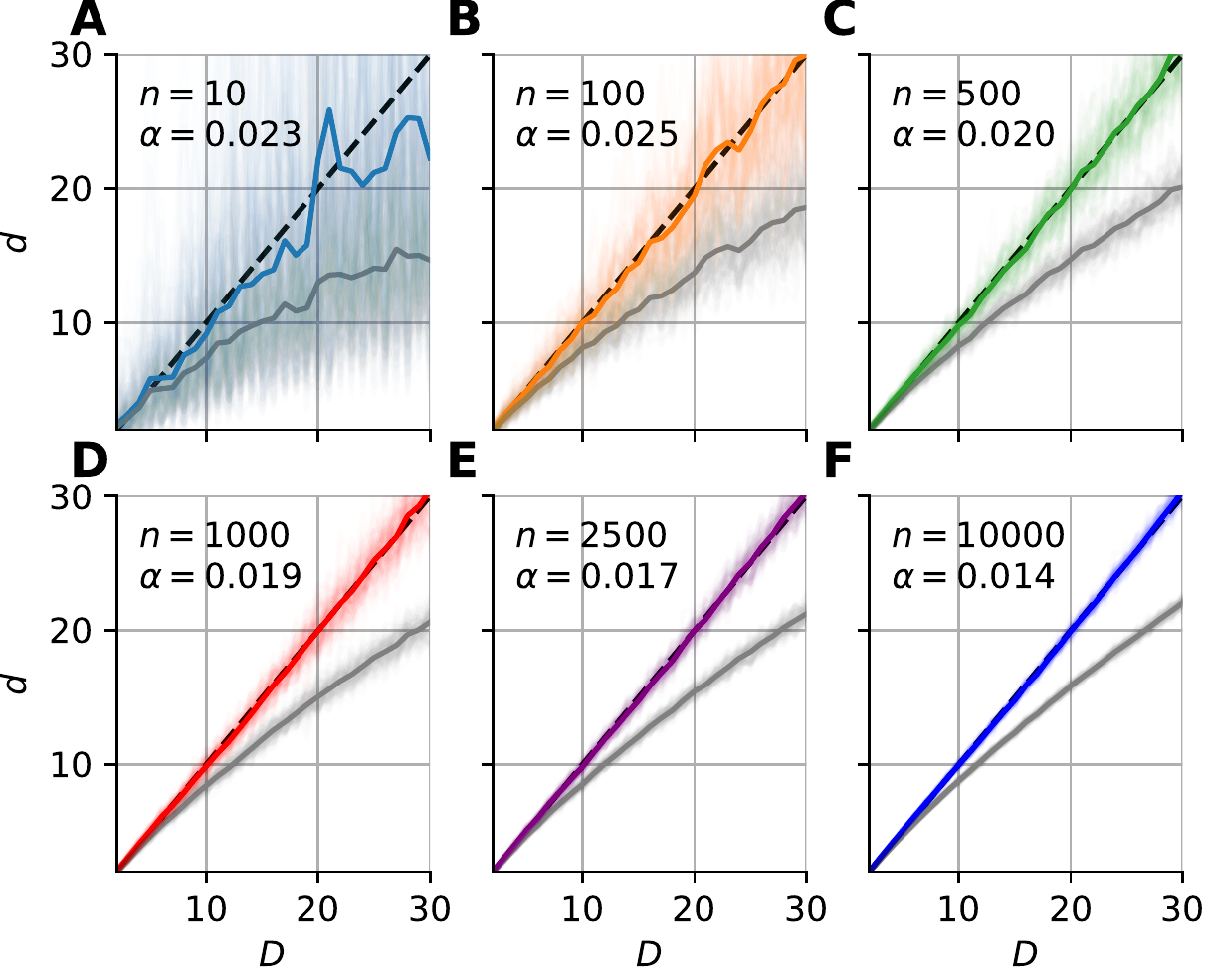}
    \caption{\textbf{Bias-correction of the median-FSA estimator for uniformly sampled unit hypercubes with various sample sizes ($k=1$).}  Subplots \textbf{A-F} show the mean of median-FSA estimator (grey line) values from $N=100$ realizations (shading) of uniformly sampled unit hypercubes. The boundary condition is hard, so the edge effect makes under-estimation more severe. The colored lines show the corrected estimates according to the $\hat{w}_c = \hat{w} \exp(\alpha \hat{w})$. }
    \label{fig:szepes_edgecorrect}
\end{figure}

\subsection*{Results on synthetic benchmarks}
We tested the mFSA estimator and its corrected version on synthetic benchmark datasets\citep{Hein2005, campadelli2015intrinsic}.
We simulated $N=100$ instances of $15$ manifolds ($M_i$, $n=2500$) with various intrinsic dimensions (see Table\,1,\,2,\,4 in Campadelli et al.\citep{campadelli2015intrinsic}, \href{http://www.mL.uni-saarland.de/code/IntDim/IntDim.htm}{http://www.mL.uni-saarland.de/code/IntDim/IntDim.htm}).

We estimated the intrinsic dimensionality of each sample and computed the mean, the error rate and Mean Percentage Error (MPE) for the estimators.
We compared the mFS, cmFS, the R and the matlab implementation of DANCo, and the Levina-Bickel estimator (Table \ref{tab:synthetic}).
cmFSA and DANCo was evaluated in two modes, in a fractal-dimension mode and in an integer dimension mode.

The mFSA and the Levina-Bickel estimator underestimated intrinsic dimensionality, especially in the cases when the data had high dimensionality.

In contrast, the cmFSA (cmFSA) estimator found the true intrinsic dimensionality of the datasets, it reached the best overall error rate ($0.277$) and 2nd best MPE (Fig.\,\ref{fig:error}, Table\,\ref{tab:synthetic}).
In some cases, it slightly over-estimated the dimension of test datasets.
Interestingly, DANCo showed implementation-dependent performance, the matlab algorithm showed the 2nd beast error rate ($0.323$) and the best MPE value (Table\,\ref{tab:synthetic}).
The R version overestimated the dimensionality of datasets in most cases.

\begin{figure}[htb!]
    \includegraphics[width=\textwidth]{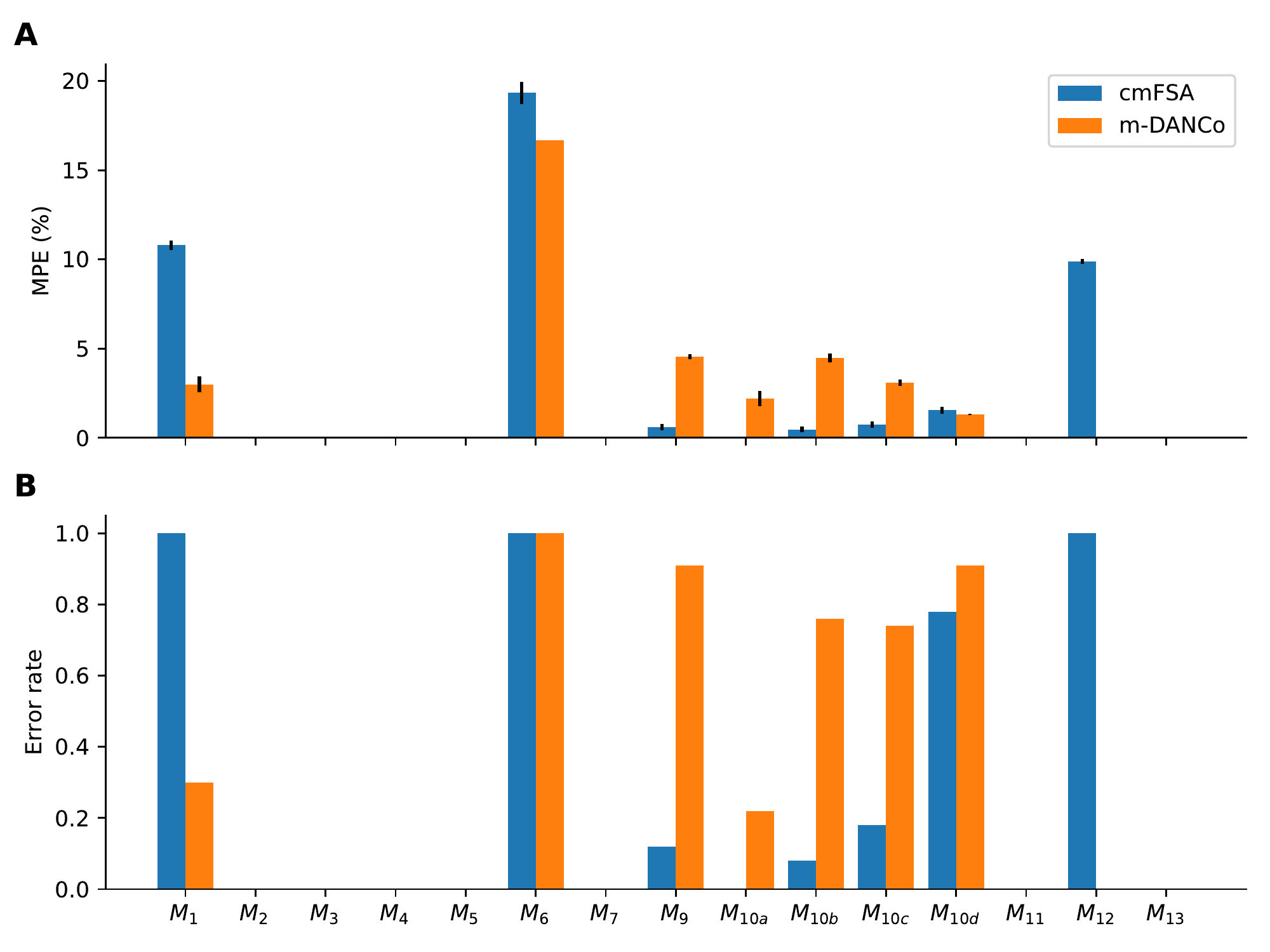}
    \caption{\textbf{Performance-comparison between cmFSA and DANCo on synthetic benchmark datasets.}
    \textbf{A} Dataset-wise Mean Percentage Error (MPE) on benchmark data. cmFSA (blue) shows smaller MPE in $4$ cases and bigger MPE in $4$ cases compared with DANCo (matlab).  
    \textbf{B} Dataset-wise error rate for cmFSA and DANCo. cmFSA shows smaller error rates in $5$ cases and bigger error rates in $2$ cases compared with DANCo.
    }
    \label{fig:error}
\end{figure}

\begin{table}[htb!]
    \centering
    \caption{\newline\textbf{Dimension estimates on synthetic benchmark datasets.}
    \newline
    The table shows true dimension values, median-Farahmand-Szepesvári-Audibert, Maximum Likelihood, corrected median Farahmand-Szepesvári-Audibert and DANCo mean estimates from $N=100$ realizations.
    The MPE values can be seen in the bottom line, the matlab version of DANCo estimator produced the smallest error followed by the cmFSA estimator.}
    \begin{tabular}{llrrrrrrrr}
\toprule
{} &  dataset &   d &    mFSA &  cmFSA$_{fr}$ &   cmFSA &  R-DANCo &  M-DANCo$_{fr}$ &  M-DANCo &  Levina \\
\midrule
1  &      $M_1$ &  10 &   9.09 &      11.19 &  11.08 &    12.00 &         10.42 &    \winner{10.30} &    9.40 \\
2  &      $M_2$ &   3 &   2.87 &       3.02 &   3.00 &     3.00 &          2.90 &     3.00 &    2.93 \\
3  &      $M_3$ &   4 &   3.83 &       4.14 &   4.00 &     5.00 &          3.84 &     4.00 &    3.86 \\
4  &      $M_4$ &   4 &   3.95 &       4.29 &   4.00 &     5.00 &          3.92 &     4.00 &    3.92 \\
5  &      $M_5$ &   2 &   1.97 &       2.00 &   2.00 &     2.00 &          1.98 &     2.00 &    1.99 \\
6  &      $M_6$ &   6 &   6.38 &       7.38 &   7.16 &     9.00 &          6.72 &     7.00 &    \winner{5.93} \\
7  &      $M_7$ &   2 &   1.95 &       1.98 &   2.00 &     2.00 &          1.96 &     2.00 &    1.98 \\
8  &      $M_9$ &  20 &  14.58 &      \winner{20.07} &  20.10 &    19.13 &         19.24 &    19.09 &   15.56 \\
9  &  $M_{10a}$ &  10 &   8.21 &       9.90 &  \winner{10.00} &    10.00 &          9.56 &     9.78 &    8.64 \\
10 &  $M_{10b}$ &  17 &  12.76 &      16.95 &  \winner{16.96} &    16.01 &         16.39 &    16.24 &   13.60 \\
11 &  $M_{10c}$ &  24 &  16.80 &      24.10 &  \winner{24.06} &    23.15 &         23.39 &    23.26 &   18.05 \\
12 &  $M_{10d}$ &  70 &  35.64 &      69.84 &  \winner{69.84} &    71.52 &         71.00 &    70.91 &   40.12 \\
13 &   $M_{11}$ &   2 &   1.97 &       2.00 &   2.00 &     2.00 &          1.97 &     2.00 &    1.98 \\
14 &   $M_{12}$ &  20 &  15.64 &      21.96 &  21.98 &    21.03 &         20.88 &   \winner{20.00} &   17.26 \\
15 &   $M_{13}$ &   1 &   1.00 &       0.96 &   1.00 &     1.00 &          1.00 &     1.00 &    1.00 \\
\bottomrule
& MPE & & 13.58 & 4.73 & 2.89 & 10.07 & 3.39 & 2.35 & 10.81 
\end{tabular}

    \label{tab:synthetic}
\end{table}

\subsection*{Analysing epileptic seizures}
To show how mFSA works on real-world noisy data, we applied it to human neural recordings of epileptic seizures.

We acquired field potential measurements from a patient with drug-resistant epilepsy by 2 electrode grids and 3 electrode strips.
We analyzed the neural recordings during interictal periods and during epileptic activity to map possible seizure onset zones (see Methods).

We found several characteristic differences in the dimension patterns between normal and control conditions.
In interictal periods (Fig.\,\ref{fig:memo_dims}\,A), we found the lowest average dimension value at the FbB2 position on the froto-basal grid.
Also, we observed a diagonal gradient of intrinsic dimensions on the cortical grid (Gr).
In contrast, we observed the lowest dimension values at the hippocampal electrode strip (JT), and the gradient on the cortical grid dissappeared during seizures (Fig.\,\ref{fig:memo_dims}\,B).
Curiously, the intrinsic dimensionality became higher at fronto-basal recording sites during seizure (Fig.\,\ref{fig:memo_dims}\,C).

\begin{figure}[t!]
    \centering
    \includegraphics[width=0.9\textwidth]{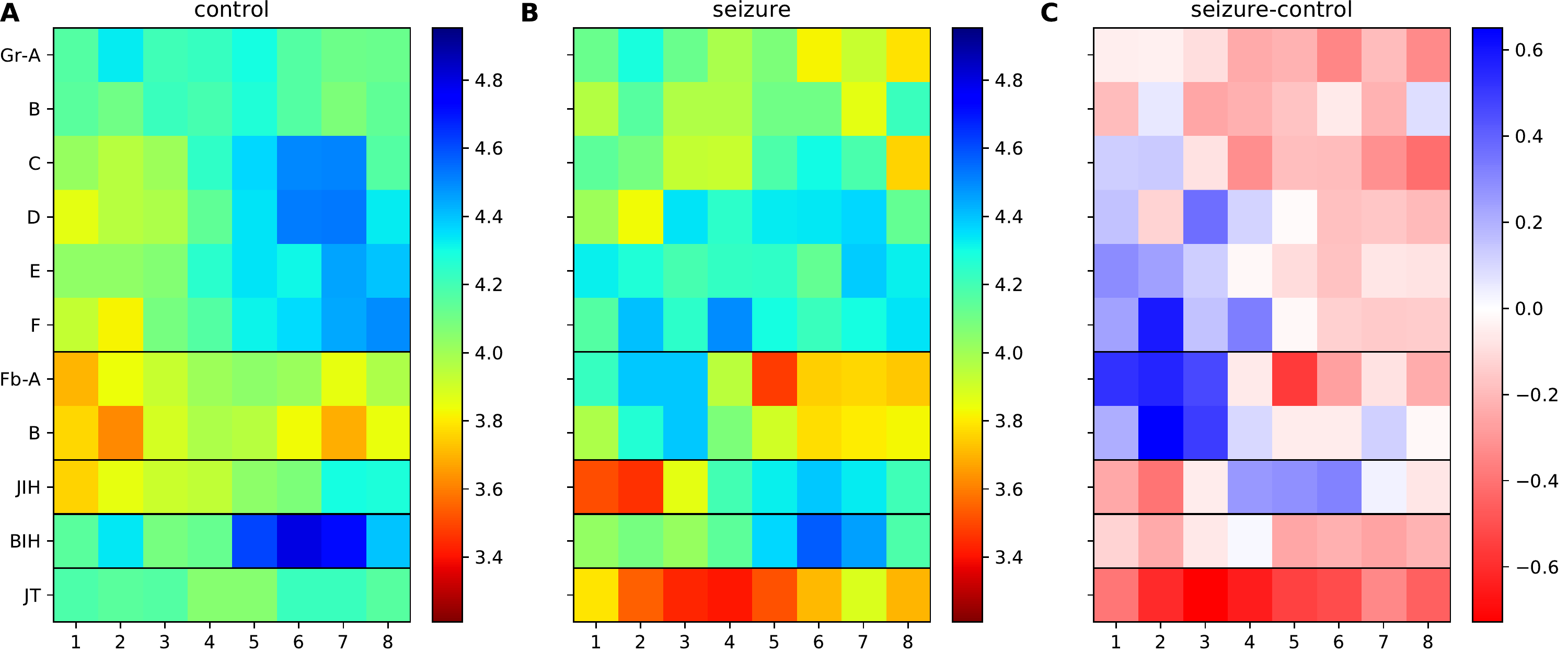}
    \caption{\textbf{mFSA Dimension estimates on intracranial Brain-LFP measurements during interictal activity and epileptic seizures.}
    The figure shows the dimension estimates on an intracranial cortical grid (Gr A-F), a smaller Frontobasal grid (Fb A, B) and 3 electrode strips with hippocampal and temporal localization (JIH, BIH, JT). The areas with lower-dimensional dynamics are marked by stronger colors.
    \textbf{A} Average of mFSA dimension values from interictal LFP activity (N=16, k=10-20).
    \textbf{B} Average of mFSA dimension values from seizure LFP activity (N=18, k=10-20).
    \textbf{C} Difference of dimension values.
    Stronger red color marks areas, where the dynamics during seizure was smaller-dimensional than its interictal counterpart.
    However, stronger blue indicates electrodes, where the during-seizure dynamics was higher dimensional than the interictal dynamics.
    }
    \label{fig:memo_dims}
\end{figure}

\section*{Discussion}
In this work we revisited and improved the manifold adaptive FSA dimension estimator.
We computed the probability density function of local estimates if the local density was uniform.
From the pdf, we derive the maximum likelihood formula for intrinsic dimensionality.

We proposed to use the median of local estimates as a global measure of intrinsic dimensionality, and demonstrated that this measure is asymptotically unbiased.

We tackled edge effects with a correction formula calibrated on hypercube datasets.
We showed that the coefficients are sample-size dependent.
Camastra and Vinciarelli \citep{Camastra2002} took a resembling empirical approach, where they corrected correlation dimension estimates with a perceptron, calibrated on d-dimensional datasets.
Our approach is different, because we tried to grasp the connection between underestimation and intrinsic dimensionality more directly, by showing that the dimension-dependence of the relative error is exponential.
The calibration procedure of DANCo may generalize better, because it compares the full distribution of local estimates rather than just a centrality measure\citep{Ceruti2014}.
Also, we are aware that our simple correction formula overlooks the effect of curvature and noise.
We tried to address the former with the choice of minimal neighborhood size ($k=1$), thus the overestimation effect due to curvature is minimal.
Additionally, the effect of noise on the estimates is yet to be investigated. There are several strategies to alleviate noise effects such as undersample the data while keeping the neighborhood fixed\citep{Facco2017}, or using a bigger neighborhood size , while keeping the sample size fixed.
Both of these procedures make the effect of curvature more severe, which makes the dimension estimation of noisy curved data a challenging task.

We benchmarked the new mFSA and corrected-mFSA method against Levina-Bickel estimator and DANCo on synthetic benchmark datasets and found that cmFSA showed comparable performance to DANCo.
For many datasets, R-DANCo overestimated the intrinsic dimensionality, which is most probably due to rough default calibration\citep{Johnsson2015}; the matlab implementation showed the best overall results in agreement with Campadelli et al\citep{campadelli2015intrinsic}.
This superiority was however dataset-specific: cmFSA performed genuinely the best in $4$, DANCo in $2$ out of the 15 benchmark datasets (with 7 ties, Table\,\ref{tab:synthetic}).
Also, cmFSA showed better overall error rate than DANCo.
Combining the performance measured by different metrics, we recognise that cmFSA found the true intrinsic dimension of the data in more cases, but when mistaken,  it makes relatively bigger errors compared with DANCo.

The mFSA algorithm revealed diverse changes in the neural dynamics during epileptic seizures.
In normal condition, the gradient of dimension values on the cortical grid reflects the hierarchical organization of neocortical information processing\citep{Tajima2015}.
During seizures, this pattern becomes disrupted pointing to the breakdown of normal activation routes. 
Some channels showed lower dimensional dynamics during seizures; that behaviour is far from the exception: the decrease in dimensionality is due to widespread synchronization events  between neural populations\citep{Mormann2000}, a phenomenon reported by various authors \citep{Polychronaki2010, Bullmore1994, Paivinen2005}.
These lower-dimensional areas are possible causal sources\citep{Sugiyama2013, Krakovska2019, Benko2018} and candidates for being the seizure onset zone.
Interestingly, Esteller et al found, that the Higuchi fractal dimension values were higher at seizure onset and decreased to lower values as the seizures evolved over time\citep{Esteller1999}. 
We found, that most areas showed decreased dimensionality, but few areas also showed increased dimension values as seizure takes place.
This may suggests that new - so far unused - neural circuits are activated at seizure onset; whether this circuitry contributes to or counteracts epileptic seizure is unclear.

\section*{Methods}

The simulations and the FSA algorithms were implemented in python3\citep{python3} using the numpy\citep{oliphant2006guide}, scipy\citep{2020SciPy-NMeth} and matplotlib\citep{hunter2007matplotlib} packages, unless otherwise stated. 

\subsection*{Simulations}

We generated test-datasets by uniform random sampling from the unit $D$-cube to demonstrate, that theoretical derivations fit to data.
We measured distances with a circular boundary condition to avoid edge effects, hence the data is as close to the theoretical assumptions as possible.

To illustrate the probability density function of the FSA estimator, we computed the local FSA intrinsic dimension values (Fig.\,\ref{fig:szepes_pdf}).
We generated $d$-hypercubes ($n=10000$, one realization) with dimensions of $2$, $3$, $5$, $8$, $10$ and $12$, then computed histograms of local FSA estimates for three neighborhood sizes: $k=1$, $11$, $50$ respectively (Fig.\,\ref{fig:szepes_pdf}\,A-F).
We drew the theoretically computed pdf to illustrate the fit.

To show that the theoretically computed sampling distribution of the mFSA fits to the hypercube datasets, we varied the sample size ($n=11, 101, 1001$) with $N=5000$ realizations from each.
We computed the mFSA for each realization and plotted the results for $d=2$ (Fig.\,\ref{fig:median_pdf}\,A) and $d=5$ (Fig.\,\ref{fig:median_pdf}\,B).

We investigated the dimensionality and sample-size effects on mFSA estimates ( Fig.\,\ref{fig:szepes_ddep}\,A-F).
We simulated the hypercube data in the $2$-$30$ dimension-range, and applied various sample sizes: $n=10, 100, 1000, 2500, 10000$; one hundred realizations each ($N=100$).
We computed the mFSA values with minimal neighborhood size ($k=1$), and observed finite-sample-effects, and asymptotic convergence.
The finite sample effect was pronounced at low sample sizes and high dimensions, but we experienced convergence to the real dimension value as we increased sample size.
We repeated the analysis with hard boundary conditions.

We fitted a correction formula on the logarithm of dimension values and estimates with the least squares method (Eq.\,\ref{eq:alpha}), using all $100$ realizations for each sample sizes separately.
\begin{equation}\label{eq:alpha}
    \alpha = \frac{\sum (\ln{E_i}) d^{(i)}}{\sum \left(d^{(i)}\right)^2}
\end{equation}
Where $E_i = D_i / d^{(i)}$ is the relative error, $D_i$ is the intrinsic dimension of the data, and $d^{(i)}$ are the corresponding mFSA estimates.
This procedure fit well to data in the intrinsic dimension range 2-30 (Fig.\,\ref{fig:szepes_edgecorrect}\,A-F).

Wider range of intrinsic dimension values (2-80) required more coefficients in the polynomial fit procedure (SFig.\,\ref{fig:calibration}\,A).
Also, we used orthogonal distance regression to fit the mean over realizations of $\ln E_i$ with the same $D_i$ value (SFig.\,\ref{fig:calibration}\,B).
We utilized the mean and standard deviation of the regression error to compute the error rate of cmFSA estimator, if the error-distributions are normal (SFig.\,\ref{fig:calibration}\,C-D). 
We applied this calibration procedure ($n=2500$) to compute cmFSA on the following benchmark datasets. 

\subsection*{Comparison on synthetic benchmark datasets}
We simulated $N=100$ instances of $15$ manifolds ($M_i$, $n=2500$) with various intrinsic dimensions (see Table\,1,\,2,\,4 in Campadelli et al.\citep{campadelli2015intrinsic}, \href{http://www.mL.uni-saarland.de/code/IntDim/IntDim.htm}{http://www.mL.uni-saarland.de/code/IntDim/IntDim.htm}).

We measured the performance of the mFSA and corrected-mFSA estimators on the benchmark datasets, and compared them with the performance of ML\citep{Levina2005} and DANCo\citep{Ceruti2014} estimators.
We used the matlab\citep{MATLAB2020, Lombardi2020}(\href{https://github.com/cran/intrinsicDimension}{https://github.com/cran/intrinsicDimension}) and an R package\citep{Johnsson2015} implementation of DANCo.

To quantify the performance we adopted the Mean Percentage Error (MPE, Eq.\ref{eq:mpe}) metric\citep{campadelli2015intrinsic}:
\begin{equation}\label{eq:mpe}
    \mathrm{MPE} = \frac{100}{M N} \sum_{j=1}^{M} \sum_{i=1}^{N} \frac{|D_j-\hat{d}_{ij}|}{D_j}
\end{equation}
Where there is $N$ realizations of $M$ types of manifolds, $d_j$ are the true dimension values, $\hat{d}_{ij}$ are the dimension estimates.

Also, we used the error rate -- the fraction of cases, when the estimator did not find (missed) the true dimensionality -- as an alternative metric.

We found that the corrected-mFSA estimator produced the second smallest MPE and the smallest error rate on the test datasets (Fig. \ref{fig:error}).

\subsection*{Dimension estimation of interictal and epileptic dynamics}

We used data of intracranial field potentials from two subdural grids positioned -- parietofrontally and frontobasally -- on the brain surface and from three strips located in the left and the right hippocampus and in the right temporal cortex as part of presurgical protocol for a subject with drug resistant epilepsy.
This equipment recorded extracellular field potentials at $88$ neural channels at a sampling rate of 2048 Hz.
Moreover, we read in -- using the neo package\citep{neo14}-- selected $10$ second long chunks of the recordings from interictal periods ($N=16$) and seizures ($N=18$) to further analysis.

We standardised the data series and computed the Current Source Density (CSD) as the second spatial derivative of the recorded potential.
We rescaled the $10$ second long signal chunks by subtracting the mean and dividing by the standard deviation.
Then, we computed the CSD of the signals by applying the graph Laplacian operator on the time-series.
The Laplacian contains information about the topology of the electrode grids, to encode this topology, we used von Neumann neighborhood in the adjacency matrix. After CSD computation, we bandpass-filtered the CSD signals\citep{Gramfort2013} (1-30 Hz, fourth order Butterworth filter) to improve signal to noise ratio. 

We embedded CSD signals and subsampled the embedded time series.
We used an iterative manual procedure to optimize embedding parameters (SFig. \ref{fig:memo_embed}).
Since the fastest oscillation is (30 Hz) in the signals, a fixed value with one fourth period ($2048 / 120 \approx 17$ samples) were used as embedding delay.
We inspected the average space-time separation plots of CSD signals to determine a proper subsampling, (with embedding dimension of D=2 (Fig. \ref{fig:memo_dims}\,A).
We found, that the first local maximum of the space-time separation was at around $5$ ms: $9-10$, $10-11$, $11-12$ samples for the $1 \%$, $25 \%$, $50\%$ percentile contour-curves respectively. 
Therefore, we divided the embedded time series into 10 subsets to ensure the required subsampling.
Then, we embedded the CSD signal up to $D=12$ and measured the intrinsic dimensionality for each embeddings (Fig. \ref{fig:memo_dims}\,B\,C).
We found that intrinsic dimension estimates started to show saturation at $D>=3$, therefore we chose $D=7$ as a sufficiently high embedding dimension (averaged over $k=10-20$ neighborhood sizes).

We measured the intrinsic dimensionality of the embedded CSD signals using the mFSA method during interictal and epileptic episodes (Fig.  \ref{fig:memo_dims}).
We selected the neighborhood size between $k=10$ and $k=20$ and averaged the resulting estimates over the neighborhoods and subsampling realizations.
We investigated the dimension values (Fig. \ref{fig:memo_dims}\,A\,B) and differences (Fig. \ref{fig:memo_dims}\,C) in interictal and in epileptic periods.

We found characteristic changes in the pattern of intrinsic dimensions during seizures, which may help to localize seizure onset zone.

\section*{Acknowledgments}
We are grateful for Ádám Zlatniczki for his comments on the manuscript.

\section*{Author contributions}
Zsigmond Benkő performed the analytical and numerical calculations and wrote the manuscript.

Marcell Stippinger corrected analytical calculations, wrote python code for numerical calculations and corrected the manuscript. 

Roberta Rehus carried out exploratory data analysis and proofreading. 

Dániel Fabó, Boglárka Hajnal, Loránd Erőss recorded the EEG data, helped with data analysis and contributed to the manuscript text.

Attila Bencze and András Telcs had profound effect on the mFSA derivations and contributed to the manuscript.

Zoltán Somogyvári led the research, helped to interpret the results of data analysis and contributed to the text.

\section*{Funding}
The research reported in this paper was supported by the BME NC TKP2020 grant of NKFIH Hungary, by the BME-Artificial Intelligence FIKP grant of EMMI (BME FIKP-MI/SC), by the National Brain Research Program of Hungary (NAP-B, KTIA\_NAP\_12-2-201), by the National Brain Project  II, NRDIO Hungary, PATTERN Group, and by 2017-1.2.1-NKP-2017-00002 of NKFIH.

\bibliographystyle{unsrt}
\bibliography{dimension}

\begin{thebibliography}{10}

\bibitem{Grassberger1983}
Peter Grassberger and Itamar Procaccia.
\newblock {Measuring the strangeness of strange attractors}.
\newblock {\em Physica D: Nonlinear Phenomena}, 1983.

\bibitem{Houle2018177}
M~E Houle, E~Schubert, and A~Zimek.
\newblock {On the Correlation Between Local Intrinsic Dimensionality and
  Outlierness}.
\newblock {\em Lecture Notes in Computer Science (including subseries Lecture
  Notes in Artificial Intelligence and Lecture Notes in Bioinformatics)}, 11223
  LNCS:177--191, 2018.

\bibitem{Dlask2017}
Martin Dlask and Jaromir Kukal.
\newblock {Correlation Dimension Estimation from EEG Time Series for Alzheimer
  Disease Diagnostics}.
\newblock In {\em Proceedings of the International Conference on Bioinformatics
  Research and Applications 2017 - ICBRA 2017}, pages 62--65. ACM Press, 2017.

\bibitem{Polychronaki2010}
G.~E. Polychronaki, P.~Y. Ktonas, S.~Gatzonis, A.~Siatouni, P.~A. Asvestas,
  H.~Tsekou, D.~Sakas, and K.~S. Nikita.
\newblock {Comparison of fractal dimension estimation algorithms for epileptic
  seizure onset detection}.
\newblock {\em Journal of Neural Engineering}, 7(4), 2010.

\bibitem{Sharma2017}
Manish Sharma, Ram~Bilas Pachori, and U.~{Rajendra Acharya}.
\newblock {A new approach to characterize epileptic seizures using analytic
  time-frequency flexible wavelet transform and fractal dimension}.
\newblock {\em Pattern Recognition Letters}, 94:172--179, jul 2017.

\bibitem{Acharya2013}
U.~Rajendra Acharya, S.~Vinitha Sree, G.~Swapna, Roshan~Joy Martis, and
  Jasjit~S. Suri.
\newblock {Automated EEG analysis of epilepsy: A review}.
\newblock {\em Knowledge-Based Systems}, 45:147--165, jun 2013.

\bibitem{Sugiyama2013}
Mahito Sugiyama and Karsten~M. Borgwardt.
\newblock {Measuring statistical dependence via the mutual information
  dimension}.
\newblock {\em IJCAI International Joint Conference on Artificial
  Intelligence}, pages 1692--1698, 2013.

\bibitem{Romano2016}
Simone Romano, Oussama Chelly, Vinh Nguyen, James Bailey, and Michael~E. Houle.
\newblock {Measuring dependency via intrinsic dimensionality}.
\newblock In {\em 2016 23rd International Conference on Pattern Recognition
  (ICPR)}, number~4, pages 1207--1212. IEEE, dec 2016.

\bibitem{Benko2018}
Zsigmond Benkő, {\'{A}}d{\'{a}}m Zlatniczki, Marcell Stippinger, D{\'{a}}niel
  Fab{\'{o}}, Andr{\'{a}}s S{\'{o}}lyom, Lor{\'{a}}nd Erőss, Andr{\'{a}}s
  Telcs, and Zolt{\'{a}}n Somogyv{\'{a}}ri.
\newblock {Complete Inference of Causal Relations between Dynamical Systems}.
\newblock {\em arXiv}, aug 2018.

\bibitem{Krakovska2019}
Anna Krakovsk{\'{a}}.
\newblock {Correlation dimension detects causal links in coupled dynamical
  systems}.
\newblock {\em Entropy}, 21(9), 2019.

\bibitem{campadelli2015intrinsic}
P~Campadelli, E~Casiraghi, C~Ceruti, and A~Rozza.
\newblock {Intrinsic Dimension Estimation: Relevant Techniques and a Benchmark
  Framework}.
\newblock {\em Mathematical Problems in Engineering}, 2015:1--21, 2015.

\bibitem{Levina2005}
Elizaveta Levina and Peter~J. Bickel.
\newblock {Maximum likelihood estimation of intrinsic dimension}.
\newblock In {\em Advances in Neural Information Processing Systems}. Neural
  information processing systems foundation, 2005.

\bibitem{Ghahramani2005}
Zoubin Ghahramani and David Mckay.
\newblock {Comments on 'Maximum likelihood estimation of intrinsic dimension'},
  2005.

\bibitem{Gupta2010}
Mithun~Das Gupta and Thomas~S. Huang.
\newblock {Regularized maximum likelihood for intrinsic dimension estimation}.
\newblock {\em Uai}, 1(1), 2010.

\bibitem{Rozza2012}
A~Rozza, G~Lombardi, C~Ceruti, E~Casiraghi, and P~Campadelli.
\newblock {Novel high intrinsic dimensionality estimators}.
\newblock {\em Machine Learning}, 89(1-2):37--65, 2012.

\bibitem{Bassis2015}
S~Bassis, A~Rozza, C~Ceruti, G~Lombardi, E~Casiraghi, and P~Campadelli.
\newblock {A Novel Intrinsic Dimensionality Estimator Based on Rank-Order
  Statistics}.
\newblock In Francesco Masulli, Alfredo Petrosino, and Stefano Rovetta,
  editors, {\em Clustering High--Dimensional Data}, pages 102--117, Berlin,
  Heidelberg, 2015. Springer Berlin Heidelberg.

\bibitem{Ceruti2014}
Claudio Ceruti, Simone Bassis, Alessandro Rozza, Gabriele Lombardi, Elena
  Casiraghi, and Paola Campadelli.
\newblock {DANCo: An intrinsic dimensionality estimator exploiting angle and
  norm concentration}.
\newblock {\em Pattern Recognition}, 47(8):2569--2581, 2014.

\bibitem{Facco2017}
Elena Facco, Maria {D 'errico}, Alex Rodriguez, and Alessandro Laio.
\newblock {Estimating the intrinsic dimension of datasets by a minimal
  neighborhood information}.
\newblock {\em Scientific Reports}, (August):1--8, 2017.

\bibitem{Camastra2002}
Francesco Camastra and Alessandro Vinciarelli.
\newblock {Estimating the intrinsic dimension of data with a fractal-based
  method}.
\newblock {\em IEEE Transactions on Pattern Analysis and Machine Intelligence},
  24(10):1404--1407, oct 2002.

\bibitem{Johnsson2015}
Kerstin Johnsson, Charlotte Soneson, and Magnus Fontes.
\newblock {Low Bias Local Intrinsic Dimension Estimation from Expected Simplex
  Skewness}.
\newblock {\em IEEE Transactions on Pattern Analysis and Machine Intelligence},
  37(1):196--202, jan 2015.

\bibitem{Chelly2016}
Oussama Chelly, Michael~E. Houle, and Ken~Ichi Kawarabayashi.
\newblock {Enhanced estimation of local Intrinsic Dimensionality using
  auxiliary distances}.
\newblock {\em NII Technical Reports}, (7), 2016.

\bibitem{Amsaleg2015}
Laurent Amsaleg, Oussama Chelly, Teddy Furon, St{\'{e}}phane Girard, Michael~E
  Houle, Ken-ichi Kawarabayashi, and Michael Nett.
\newblock {Estimating Local Intrinsic Dimensionality}.
\newblock In {\em Proceedings of the 21th ACM SIGKDD International Conference
  on Knowledge Discovery and Data Mining - KDD '15}, number~Cd, pages 29--38,
  New York, New York, USA, 2015. ACM Press.

\bibitem{Amsaleg2018}
Laurent Amsaleg, Oussama Chelly, Teddy Furon, St{\'{e}}phane Girard, Michael~E.
  Houle, Ken-ichi Kawarabayashi, and Michael Nett.
\newblock {Extreme-value-theoretic estimation of local intrinsic
  dimensionality}.
\newblock {\em Data Mining and Knowledge Discovery}, 32(6):1768--1805, nov
  2018.

\bibitem{Amsaleg2019}
Laurent Amsaleg, Oussama Chelly, Michael~E. Houle, Ken-ichi Kawarabayashi,
  Milo{\v{s}} Radovanovi{\'{c}}, and Weeris Treeratanajaru.
\newblock {Intrinsic Dimensionality Estimation within Tight Localities}.
\newblock In {\em Proceedings of the 2019 SIAM International Conference on Data
  Mining}, pages 181--189. Society for Industrial and Applied Mathematics, may
  2019.

\bibitem{Farahmand2007}
Amir~Massoud Farahmand, Csaba Szepesv{\'{a}}ri, and Jean-Yves Audibert.
\newblock {Manifold-adaptive dimension estimation}.
\newblock In {\em Proceedings of the 24th international conference on Machine
  learning - ICML '07}, pages 265--272. ACM Press, 2007.

\bibitem{Laplace1986}
Pierre~Simon Laplace.
\newblock {Memoir on the Probability of the Causes of Events}.
\newblock {\em Statistical Science}, 1(3):364--378, aug 1986.

\bibitem{Hein2005}
Matthias Hein and Jean-Yves Audibert.
\newblock {Intrinsic dimensionality estimation of submanifolds in R d}.
\newblock In {\em Proceedings of the 22nd international conference on Machine
  learning - ICML '05}, pages 289--296. ACM Press, 2005.

\bibitem{Tajima2015}
Satohiro Tajima, Toru Yanagawa, Naotaka Fujii, and Taro Toyoizumi.
\newblock {Untangling Brain-Wide Dynamics in Consciousness by Cross-Embedding}.
\newblock {\em PLOS Computational Biology}, 11(11):e1004537, nov 2015.

\bibitem{Mormann2000}
Florian Mormann, Klaus Lehnertz, Peter David, and Christian {E. Elger}.
\newblock {Mean phase coherence as a measure for phase synchronization and its
  application to the EEG of epilepsy patients}.
\newblock {\em Physica D: Nonlinear Phenomena}, 144(3):358--369, 2000.

\bibitem{Bullmore1994}
E.T. Bullmore, M.J. Brammer, P.~Bourlon, G.~Alarcon, C.E. Polkey, R.~Elwes, and
  C.D. Binnie.
\newblock {Fractal analysis of electroencephalographic signals intracerebrally
  recorded during 35 epileptic seizures: evaluation of a new method for
  synoptic visualisation of ictal events}.
\newblock {\em Electroencephalography and Clinical Neurophysiology},
  91(5):337--345, nov 1994.

\bibitem{Paivinen2005}
Niina P{\"{a}}ivinen, Seppo Lammi, Asla Pitk{\"{a}}nen, Jari Nissinen, Markku
  Penttonen, and Tapio Gr{\"{o}}nfors.
\newblock {Epileptic seizure detection: A nonlinear viewpoint}.
\newblock {\em Computer Methods and Programs in Biomedicine}, 79(2):151--159,
  aug 2005.

\bibitem{Esteller1999}
R.~Esteller, G.~Vachtsevanos, J.~Echauz, T.~Henry, P.~Pennell, C.~Epstein,
  R.~Bakay, C.~Bowen, and B.~Litt.
\newblock {Fractal dimension characterizes seizure onset in epileptic
  patients}.
\newblock In {\em 1999 IEEE International Conference on Acoustics, Speech, and
  Signal Processing. Proceedings. ICASSP99 (Cat. No.99CH36258)}, pages
  2343--2346 vol.4. IEEE, 1999.

\bibitem{python3}
Guido Van~Rossum and Fred~L. Drake.
\newblock {\em Python 3 Reference Manual}.
\newblock CreateSpace, Scotts Valley, CA, 2009.

\bibitem{oliphant2006guide}
Travis~E Oliphant.
\newblock {\em A guide to NumPy}, volume~1.
\newblock Trelgol Publishing USA, 2006.

\bibitem{2020SciPy-NMeth}
Pauli {Virtanen}, Ralf {Gommers}, Travis~E. {Oliphant}, Matt {Haberland}, Tyler
  {Reddy}, David {Cournapeau}, Evgeni {Burovski}, Pearu {Peterson}, Warren
  {Weckesser}, Jonathan {Bright}, St{\'e}fan~J. {van der Walt}, Matthew
  {Brett}, Joshua {Wilson}, K.~{Jarrod Millman}, Nikolay {Mayorov}, Andrew
  R.~J. {Nelson}, Eric {Jones}, Robert {Kern}, Eric {Larson}, CJ~{Carey},
  {\.I}lhan {Polat}, Yu~{Feng}, Eric~W. {Moore}, Jake {Vand erPlas}, Denis
  {Laxalde}, Josef {Perktold}, Robert {Cimrman}, Ian {Henriksen}, E.~A.
  {Quintero}, Charles~R {Harris}, Anne~M. {Archibald}, Ant{\^o}nio~H.
  {Ribeiro}, Fabian {Pedregosa}, Paul {van Mulbregt}, and SciPy 1.~0
  {Contributors}.
\newblock {SciPy 1.0: Fundamental Algorithms for Scientific Computing in
  Python}.
\newblock {\em Nature Methods}, 17:261--272, 2020.

\bibitem{hunter2007matplotlib}
John~D Hunter.
\newblock Matplotlib: A 2d graphics environment.
\newblock {\em Computing in science \& engineering}, 9(3):90--95, 2007.

\bibitem{MATLAB2020}
MATLAB.
\newblock {\em {MATLAB version 9.8.0.1396136 (R2020a)}}.
\newblock The Mathworks, Inc., Natick, Massachusetts, 2020.

\bibitem{Lombardi2020}
Gabriele Lombardi.
\newblock Intrinsic dimensionality estimation techniques.
\newblock {\em MATLAB Central File Exchange}, Retrieved July 16, 2020.

\bibitem{neo14}
S.~Garcia, D.~Guarino, F.~Jaillet, T.R. Jennings, R.~Pröpper, P.L. Rautenberg,
  C.~Rodgers, A.~Sobolev, T.~Wachtler, P.~Yger, and A.P. Davison.
\newblock Neo: an object model for handling electrophysiology data in multiple
  formats.
\newblock {\em Frontiers in Neuroinformatics}, 8:10, February 2014.

\bibitem{Gramfort2013}
Alexandre Gramfort, Martin Luessi, Eric Larson, Denis Engemann, Daniel
  Strohmeier, Christian Brodbeck, Roman Goj, Mainak Jas, Teon Brooks, Lauri
  Parkkonen, and Matti H{\"{a}}m{\"{a}}l{\"{a}}inen.
\newblock {MEG and EEG data analysis with MNE-Python}.
\newblock {\em Frontiers in Neuroscience}, 7:267, 2013.

\end{thebibliography}

\newpage
\section*{Supplemental information}
\label{sec:appendix}
\addcontentsline{toc}{section}{\nameref{sec:appendix}}
\appendix

\renewcommand\thefigure{\arabic{figure}}
\renewcommand\figurename{SFig.}
\setcounter{figure}{0}

\renewcommand\thetable{\arabic{table}}
\renewcommand\tablename{STable}
\setcounter{table}{0}

\renewcommand{\theequation}{S.\arabic{equation}}
\setcounter{equation}{0} 

\section{Calculations for normalized distances}
\subsection{Distance density of the $k$-th nearest neighbor}\label{si:r_pdf}
Let's take $K-1$ points in the unit $D$-sphere randomly, and we chose one with $r$ distance from the center.
This situation simulates a $K$-neighborhood, with normalized distances of $K-1$ points from the center.
The next formula tells us the probability that a selected point at $r$ was the $k$th from the center.

\begin{equation}
    P(k|r, K, D) = \binom{K-2}{k-1} r ^ {D(k-1)} (1 - r^D) ^ {K-k-1}
\end{equation}
here $r$ can take values from the $[0, 1]$ interval.

Moreover the probability density that there is a point at $r$ radius is given by the following derivation formula:

\begin{equation}\label{eq:prifd}
   p(r| D) = D r ^ {D-1}
\end{equation}
If sampling process is independent, the pdf that a point is on the radius $r$ from $K-1$ points is the same and independent of sample size:

\begin{equation}
    \begin{aligned}
    p(r|K-1, D) &= \sum_{j=1}^{n} \frac{1}{n}  \underbrace{\int \mathrm{d}r_1 \dots \int \mathrm{d}r_i \dots \int \mathrm{d}r_n}_{i \neq j}  p(r_1, r_2, \dots r_j=r, \dots r_n | D) \\
&= \sum_{j=1}^{n} \frac{1}{n}  \underbrace{\int \int \dots \int}_{n-1}  \prod_{i=1}^n D {r_i}^ {D-1} \underbrace{\mathrm{d}r_i}_{i \neq j} 
= \frac{1}{n} \sum_{j=1}^{n}  D {r_j}^ {D-1}
= D r ^ {D-1}
    \end{aligned}
\end{equation}

This is the prior pdf of distance, we assume uniform density in the n-sphere.
This prior can be any density, we chose this specific form with respect to the maximum entropy principle and also for practical reasons.

From the previous two formulas, we can write up the joint mixed probability function:

\begin{equation}
    p(k, r | K-1, D) = D \binom{K-2}{k-1} r ^ {Dk-1} (1 - r^D) ^ {n-k}
\end{equation}

Also:

\begin{equation}
    p(k|K-1, D) = \frac{1}{K-1}
\end{equation}

Using Bayes theorem, we derive the distance distribution of the $k$th neighbor:

\begin{align}
    p(r|k, K-1, D) &= \frac{P(k|r, K-1, D) p(r|K-1, D)}{p(k|K-1, D)}\\
    &= (K-1) D \binom{K-2}{k-1} r ^ {D k-1} (1 - r^D) ^ {K-k-1}\\
    &=  \frac{D}{B(k, K-k)}  r ^ {D k-1} (1 - r^D) ^ {K-k-1} \label{eq:pr}
\end{align}
Where $B$ is the beta function.

\subsection{Maximum Likelihood estimation of intrinsic dimension}
\label{si:ML}
Given a dataset, we can use the Maximum Likelihood principle to estimate intrinsic dimensionality by using Eq.\,\ref{eq:pr}.
The dataset is $K-1$ randomly sampled points inside a $d$-dimensional sphere.
But first we have to express the likelihood function: 

\begin{equation}
    \mathcal{L}(D | X) =  p(r_1, ..., r_{K-1}| D)
\end{equation}

This expression can be factorized into a chain because $p(r_k| r_{k+1}, r_{k+2}, ..., r_{K-1}) = p(r_k| r_{k+1})$  which is a Markov property of neighbor distances.

\begin{equation}
    \mathcal{L}(D | X) =  p(r_1, ..., r_{K-1}| D) = \prod_1^{K-1} p(r_k| r_{k+1}, D)
\end{equation}
where $r_{K}=1$.

\begin{equation}
    p(r_k| r_{k+1}, D) = k D \left( \frac{r_k}{r_{k+1}} \right)^{kD-1} \frac{1}{r_{k+1}} 
\end{equation}

So if we substitute back into the previous expression:

\begin{equation}
    \begin{split}
        \mathcal{L}(D | X) &=  p(r_1, ..., r_n| D) = \prod_1^{K-1} p(r_k| r_{k+1}, d)\\
 &= {(K-1)!} \, D^{K-1} \frac{r_1^{D-1}}{r_2^{D}} \frac{r_2^{2D-1}}{r_3^{2D}} \frac{r_3^{3D-1}}{r_4^{3D}} \dots \frac{r_{K-1}^{(K-1) D - 1}}{r_{K}^{(K-1) D}}\\
&= {(K-1)!} \, D^{K-1} \left(\prod_1^{K-1} r_k \right)^{D-1}
    \end{split}
\end{equation}

The log likelihood:

\begin{equation}
    \log{\mathcal{L}(D | X)} = \left(\sum_1^{K-1} \log{k} \right) + (K-1) \log{D} + (D-1) \sum_1^{K-1} \log{r_k}
\end{equation}

We seek for extrema of the likelihood function: 

\begin{equation}
    \begin{split}
        \frac{ \partial \log{ \mathcal{L}(D | X) }}{\partial D} \stackrel{!}{=} 0\\
        \frac{K-1}{D} + \sum_1^{K-1} \log{r_k} \stackrel{!}{=} 0 \\
    \end{split}
\end{equation}

\begin{equation}
    \boxed{ d_{\mathrm{ML}} = \frac{K-1}{-\sum_1^{K-1} \log{r_k}}}
\end{equation}

This latter formula is basically equivalent to the local  Levina-Bickel ML intrinsic dimension estimator if $r_k = \frac{R_k}{R_{K}}$.

\section{Derivation of the pdf of the FSA estimator}
The starting point of our derivation is the posterior density of $r$, computed in Section 1:

\begin{equation}
    p(r|k, K-1, D) = \frac{D}{B(k, K-k)}  r ^ {D k-1} (1 - r^D) ^ {K-k-1} 
\end{equation}

We fill in $K = 2k$ to the previous expression:

\begin{equation}
    p(r|k, 2k-1, D) = \frac{D}{B(k, k)}  r ^ {D k-1} (1 - r^D) ^ {k-1}  
\end{equation}

The pdf of $w$ can be expressed from the pdf of $r$ with simple intergal transform:

\begin{equation}
    p \left( r|k, 2k-1, D \right) \mathrm{d}r = q \left( d \right) \mathrm{d}d 
\end{equation}
so

\begin{equation}
     q \left( d \right) = p \left( r|k, 2k-1, D \right) \left| \frac{\mathrm{d}r}{\mathrm{d}d} \right|
\end{equation}

To compute the above expression, we first express $r$ as a function of $d$, then we compute the derivative. Afterwards we put the things together.

\begin{equation}
    \begin{split}
    d = -\frac{\log{2}}{\log{r}} \implies
    r= \exp{\left(-\frac{\log{2}}{d} \right)} \implies
    \frac{\mathrm{d}r}{\mathrm{d}d} = \exp{\left(-\frac{\log{2}}{d} \right)} \frac{\log{2}}{d^2}
    \end{split}
\end{equation}

And finally, we put together these parts to get the pdf of the FSA estimator (Fig. \ref{fig:szepes_pdf}):

\begin{eqnarray}
    q(d|k, D) &=& \frac{D}{B(k, k)}  e^{\left(-\frac{\log{2}}{d} (D k-1) \right)} \left(1 - e^{\left(-\frac{\log{2}}{d} D \right)} \right)^{k-1} e^{\left(-\frac{\log{2}}{d} \right)} \frac{\log{2}}{d^2}\nonumber=\\
     &=& \boxed{ \frac{D \log{(2)}}{B(k, k)}   \frac{2^{-\frac{D k}{d}} \left(1 - 2^{-\frac{D}{d} } \right)^{k-1}}{d^2} } \label{eq:szepes_pdf_s}
\end{eqnarray}

where $B(k, k)$ is the Euler beta function.

\section{Supplemental Figures and tables}

\begin{figure}
    \centering
    \includegraphics[width=\linewidth]{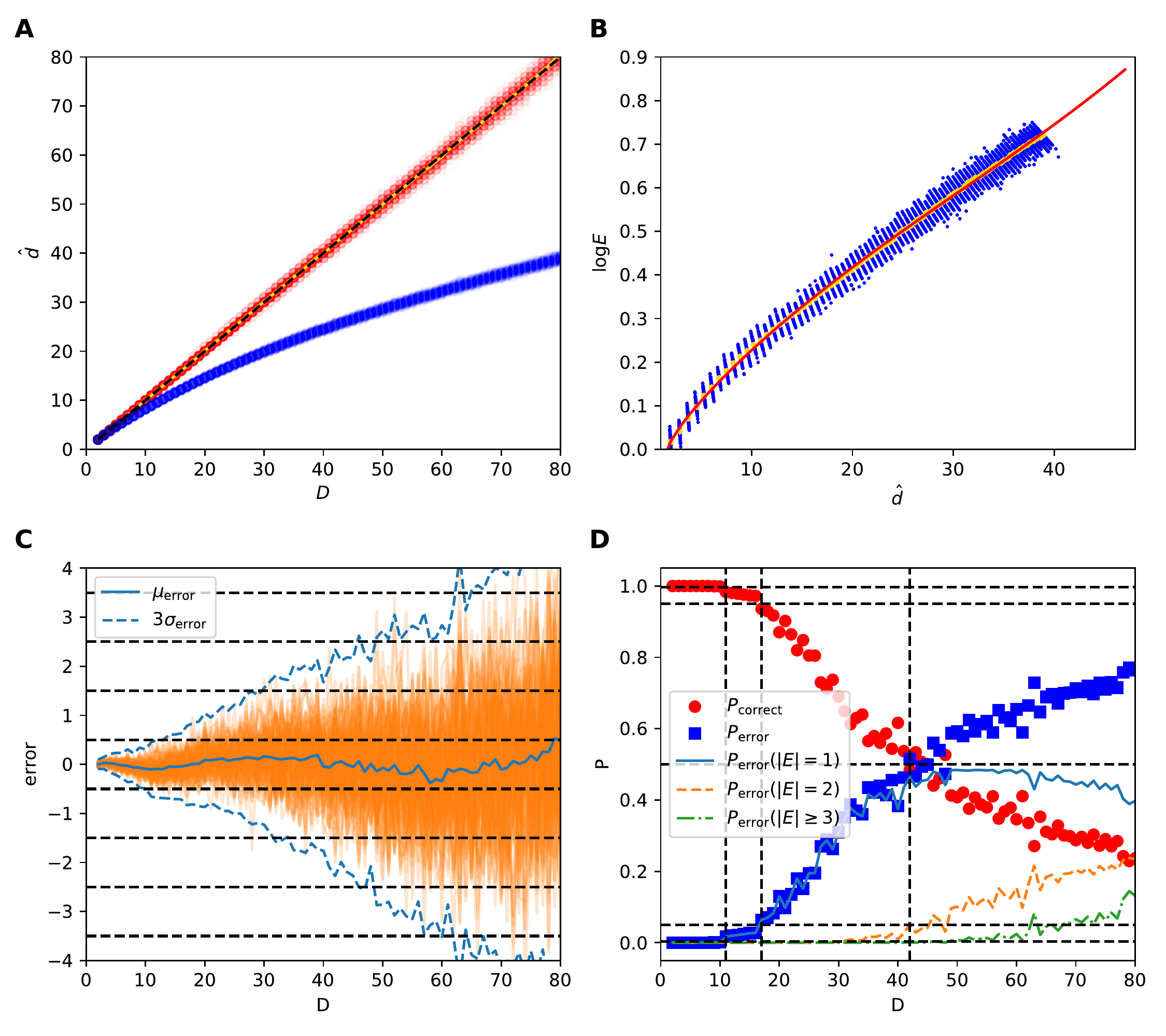}
    \caption{\textbf{Calibration procedure for the $n=2500$  datasets up to $D=80$ ($k=5$).} The figure shows the calibration procedure on $100$ instances of uniformly sampled hypercubes.
    \textbf{A} Dimension estimates in the function of intrinsic dimensionality for the calibration hypercubes. The diagonal (dashed) is the ideal value, however the mFSA estimates (blue) show saturation because of finite sample and edge effects. cmFSA estimates (red) are also shown, with the mean (yellow) almost aligned with the diagonal.
    \textbf{B} The relative error ($E$) in the function of uncorrected mFSA dimension on semilogarithmic scale. The error-mFSA pairs (blue) lie on a short stripe for each intrinsic dimension value. The subplot also shows id-wise average points (yellow) and the polynomial fitting curve (red).
    \textbf{C} The error of cmFSA estimates in the function of intrinsic dimension on the calibration datasets. The mean error (blue line) oscillates around zero and the $99.7\%$ confidence interval (blue dashed) widens as ID grows. The rounding switchpoints are also shown.
    \textbf{D} The probability that cmFSA hits the real ID of data, or misses by one, two or more as a function of ID on the calibration dataset. 
    }
    \label{fig:calibration}
\end{figure}

\begin{figure}[t!]
    \centering
    \includegraphics[width=0.9\textwidth]{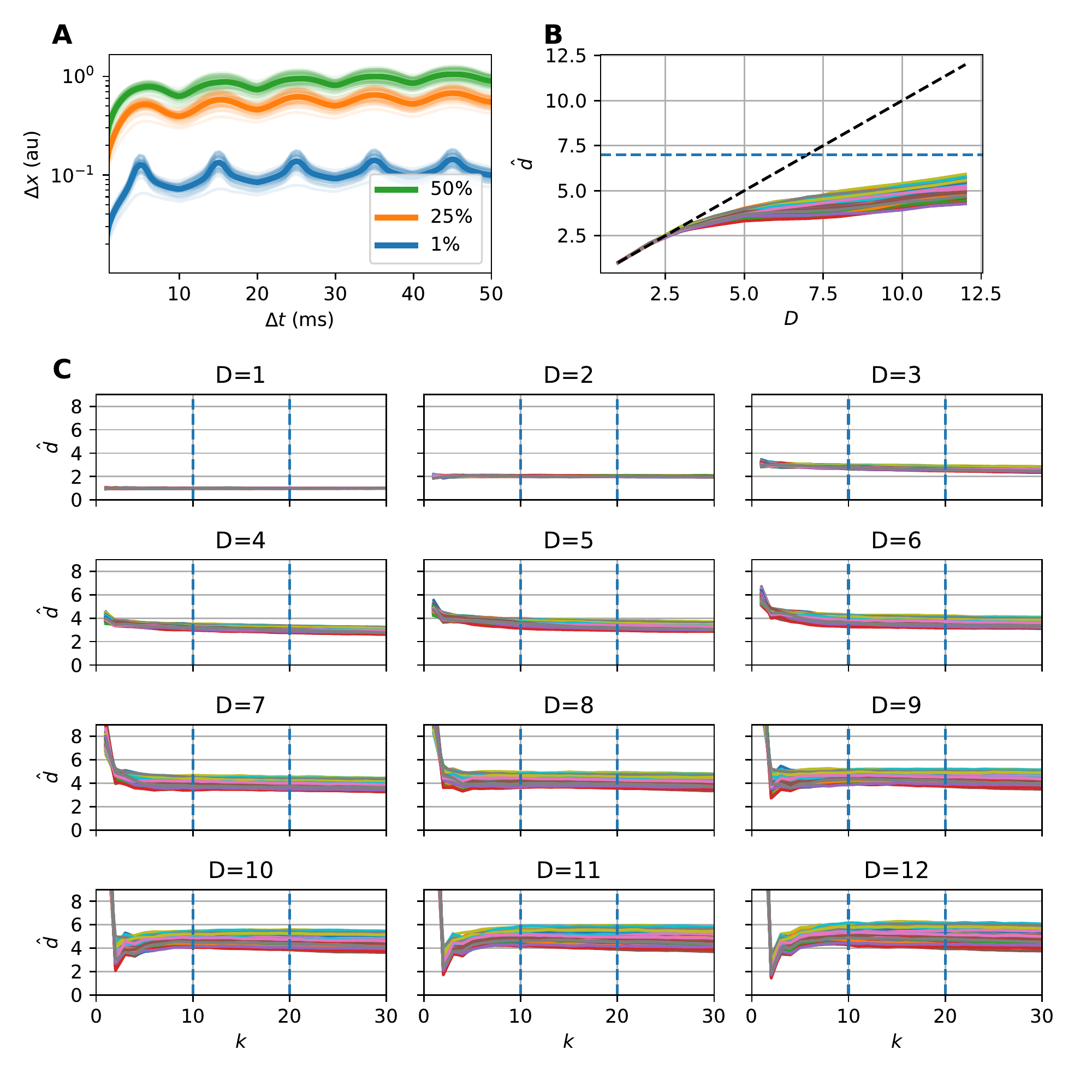}
    \caption{\textbf{Subsampling and embedding of the CSD signals.} 
    \textbf{A} Mean Space-time separation plot of the CSD recordings, the lines show the contours of the $1\%$ (blue),  $25\%$ (orange),  and $50\%$ (green) percentiles for the 34 - 16 interictal and 18 seizures - recordings (thin lines) and their average (thick line, $D=2$).
    The first local maximum is at around $5$ ms ($10$ time steps), which appoints the proper subsampling to avoid the effect of temporal correlations during the dimension estimation.
    \textbf{B} Intrinsic dimension in the function of the embedding dimension for the $88$ recording-channels (averaged between $k=5-10$, for the first seizure).
    Dimension-estimates deviate from the diagonal above $D=3$, thus we chose $D=2*3+1=7$ as embedding dimension.
    \textbf{C} Intrinsic dimension in the function of neighborhood size for various embedding dimensions (88 channels, for the first seizure). The dimension estimates are settled at the neighborhood size between k=$10-20$ (dashed blue). The knee because of the autocorrelation becomes pronounced for $D>=8$.
    }
    \label{fig:memo_embed}
\end{figure}

\begin{table}[h!]
    \flushright
    \caption{
    \textbf{Used symbols with interpretation.}}
       \begin{tabular}{rcl}
    $k$ &-& the order of the neighbor (increasing order as the distance from the center rises)\\
    $K-1$ &-& number of points in the neighborhood\\
    $R$ &-& distance from center \\
    $r$ &-& normalized distance from center $r = R / R_K$ ($r\in [0, 1]$)\\
    $\eta$ &-& local density-dependent factor, approximately independent of $R$\\
    $D$ &-& intrinsic dimensionality of the space where the points are.\\
    $d$ &-& estimated intrinsic dimension value\\
    $P$ &-& Probability, probability mass function\\
    $p$ or $q$ &-& probability density function (pdf)\\
    $n$ &-& sample size\\
    $N$ &-& number of realizations\\ 
    $B$ &-& Euler beta function\\
\end{tabular}
    \label{tab:symbols}
\end{table}

\end{document}